\documentclass[oneside]{article}
\usepackage{a4wide}
\usepackage{fullpage}
\usepackage {hyperref}
\usepackage{authblk}
\usepackage[algo2e,ruled,linesnumbered]{algorithm2e}
\SetKwComment{Comment}{/* }{ */}
\SetKwProg{Init}{Initialization:}{}{}
\usepackage{graphicx}
\usepackage{tikz}
\usetikzlibrary{matrix}
\usetikzlibrary{calc}
\usetikzlibrary{arrows, fit, matrix, positioning, shapes, backgrounds,
}
\usepackage{nicematrix}
\SetKwComment{Comment}{/* }{ */}
\usepackage{thm-restate}
\usepackage{cleveref}
\usepackage[normalem]{ulem}
\useunder{\uline}{\ul}{}
\usepackage{comment}
\usepackage{color}
\usepackage{multicol}
\usepackage{amsmath}
\newtheorem{claim}{Claim}

\newtheorem{theorem}{Theorem}
\newtheorem{lemma}[theorem]{Lemma}

\newenvironment{proof}[1][Proof]

\newcommand{\Xomit}[1]{}

\newcommand{\ssend}{{\sf{send}}}

\begin{document}
	
\title{Some New Results With $k$-Set Agreement}
	
\author[1]{Carole Delporte-Gallet}
\author[1]{Hugues Fauconnier}
\author[1,2]{Mouna Safir}
	
\affil[1]{IRIF, Univeristé Paris Cité}
\affil[2]{UM6P-CS, Université Mohammed VI Polytechnique }

\date{}
\maketitle

\begin{abstract}
In this article, we investigate the solvability of $k$-set agreement among $n$ processes in distributed systems  prone to different types of process failures. 
Specifically, we explore two scenarios: synchronous
message-passing systems prone to up to $t$  Byzantine failures of
processes.
And  asynchronous  shared memory systems prone to up to $t$ crash
failures of processes.
Our goal is to address the gaps left by previous works\cite{SSS,AsynchKset} in these areas. For Byzantine failures case we consider systems with authentication where processes have unforgeable signatures.

For synchronous message-passing systems, we present an authenticated
algorithm that achieves $k$-set agreement in only two rounds, with no
constraints on the number of faults $t$, with $k$ determined as $k \geq
\lfloor \frac{n}{n-t} \rfloor + 1$. In fact the lower bound for
$k$, for the Byzantine case  is $k \geq
\lfloor \frac{n}{n-t} \rfloor $ that is obtained by an algorithm based
on traditional consensus with $t+1$ rounds.

In  asynchronous  shared memory systems, we introduce an algorithm that accomplishes $k$-set agreement for values of $k$ greater than $ \lfloor \frac{n-t}{n-2t} \rfloor +1$. This algorithm uses a snapshot primitive to handle crash failures and enable effective set agreement.

\end{abstract}

\smallskip

\textbf{Keywords:} Byzantine failures, Crash failures, Distributed systems,  $k$-set agreement.

\section{Introduction}

The consensus problem is an abstraction of many coordination problems in a distributed system that can suffer process failures. Roughly speaking, the consensus problem is to have processes of a distributed system agree on a common decision. Because of the many practical problems that can be reduced to this simple primitive, consensus has been thoroughly studied. We refer the reader to \cite{consensus} for a detailed discussion of consensus. Motivated by the significance of consensus, researchers have explored variations of the problem to investigate the boundaries of what is possible and impossible. One such variation is the $k$-set consensus \cite{chaudhuri1993more}, which relaxes the safety conditions of consensus to allow for a set of decision values with a cardinal of up to $k$ (compared to $k=1$ in consensus). The $k$-set agreement problem has been widely studied in the field of distributed computing \cite{raynal2010fault}. Beyond the practical interest of this problem, particularly regarding fault-tolerant distributed computing, one of the main reasons behind the focus on $k$-set agreement problem is the fact that it can be used to define and compare computational power properties of systems.

In $k$-set agreement, each process must decide on a value such that no more than $k$ different values are decided by processes.
In addition, the processes must guarantee a validity condition which characterizes which decision values are allowed as a function of the input values and whether failures occur.

Hence, with crash process failures, the validity condition generally considered ensures that the decided values are initial values proposed by processes. Regarding $k$-set agreement in asynchronous models one of the most famous (and difficult) results is the extension of the consensus impossibility result to impossibility of the $k$-set agreement~\cite{BG93,HS93,saks2000wait} when at least $k$ processes may fail. In synchronous models with crash failures, $k$-set agreement is solvable for all $k$. But interestingly, an (imperfect) agreement on more than one value will divide the complexity in the number of rounds needed to solve the $k$-set agreement: as proved in ~\cite{CHLT00}, 
$\lfloor t/k \rfloor +1$ rounds of communication are necessary and sufficient to solve $k$ set agreement with no more than $t$ faulty processes. Note that these results depends on the chosen validity condition.
Note also that an important interest of the $k$-set agreement is its universality in the sense where the $k$-set agreement allows the state machine (with some liveness conditions) replication~\cite{GG11}. 

With more severe failures than crashes, initial values of faulty processes has no real meaning (what is the initial value of a Byzantine process?). Then other validity properties have been defined for the $k$-set agreement. The work presented in \cite{AsynchKset} investigates the $k$-set consensus problem, considering various problem definitions and models with both crash and Byzantine failures, and shows that the precise definition of the validity requirement is crucial to the solvability of the problem. The authors of \cite{AsynchKset}  envisage up to six validity properties. Among them the strong validity ensuring that if all correct processes have the same initial value $v$ then all correct decide $v$, is the more appropriate for Byzantine failures. 

In this work, we aim to investigate the solvability of $k$-set agreement in distributed systems prone to different types of failures with this validity condition ensuring that if all correct processes have the same initial $v$ all correct processes decide $v$.
Specifically, we focus on two scenarios: synchronous message-passing systems prone to Byzantine failures
and  asynchronous shared memory systems prone to crash failures.
Concerning Byzantine failures we consider models with authentication in which messages may be signed by processes with unforgeable signatures. Most of the results concerning general Byzantine failures are already shown in~\cite{bouzid2016necessary,SSS}. Our objective is to address gaps left by previous works in these areas and provide insights into the solvability of $k$-set agreement under specific failures models and give some results in term of rounds complexity.

For synchronous message-passing systems, the $k$-set agreement is possible to achieve only when $k \geq \lfloor \frac{n}{n-t} \rfloor$.
We present an authenticated algorithm that achieves $k$-set agreement in only two rounds, with no constraints on the number of failures $t$, but the value of $k$  is greater or equal to $\lfloor \frac{n}{n-t} \rfloor + 1$ and hence is not optimal. To achieve an optimal $k$, we propose an algorithm that spans $t+1$ rounds and guarantees $k$-set agreement for $k = \lfloor \frac{n}{n-t} \rfloor$ for any value of $t$. This algorithm leverages $n$ instances of the Terminating Reliable Broadcast (TRB), where the delivered value represents the proposed value for the set agreement. 
This result is interesting: if $k$ is the optimal value for $k$-set agreement, we have $t+1$ rounds 
to achieve the $k$-set but only 2 rounds for the $(k+1)$-set. That means for example that only 2 rounds are needed for the 2-set agreement. Note also that these results apply to crash failure models with this considered validity property.

In asynchronous shared memory  systems, we propose an algorithm that accomplishes $k$-set agreement for values of $k$ strictly greater than $\lfloor\frac{n-t}{n-2t}\rfloor$. This algorithm effectively handles crash failures using a snapshot primitive.


In summary, our work contributes to the understanding of $k$-set agreement in distributed systems, providing valuable insights into the solvability of this problem under specific failure models. We address important gaps in existing research and offer practical solutions to achieve $k$-set agreement in both synchronous and asynchronous distributed systems with different types of failures. The rest of the paper is organized as follows: Section 2 presents the model and some preliminary results as a bound on the value ok $k$ enabling $k$-set Agreement,  Section 3 presents an only two rounds algorithm for $k$-set Agreement for $k>\lfloor \frac{n}{n-t} \rfloor$ and a $t+1$ rounds algorithm for $k \geq \lfloor \frac{n}{n-t} \rfloor$ in synchronous message passing model with Byzantine failures and authentication, Section 4 presents results in asynchronous shared memory model. Finally, Section 5 concludes the paper and discusses future research directions.

\section{Preliminaries} \label{sec:prelims}
In this section, we provide a detailed explanation of the communication model and failure models used in our study, as well as an overview of two essential primitives, the Terminating Reliable Broadcast (TRB), and the Snapshot primitive. 

In the following $n$ denote the number of processes and $t$ the maximum number of faulty processes.
\subsection{Communication Model}
We consider systems with $n$ processes with at most $t$ processes may be faulty. There is no communication failure. We first describe the communication models.
\subsubsection{Message Passing Model}
In the \emph{message passing model}, we consider a system consisting of $n$ processes that communicate by sending and receiving messages over a complete point to point communication network
without communication failure.
Any message sent by a correct process is eventually received by its receiver.

In \emph{synchronous} message passing model, the messages are guaranteed to arrive within a bounded time interval $\Delta$ to their receiver. In the following, in the synchronous message passing model processes run in synchronized rounds: at each round the processes send messages that are received in the same round.
On the other hand, in \emph{asynchronous} model, there are no such timing guarantees.

\subsubsection{Shared Memory Model}
In the \emph{shared memory model}, processes communicate by reading from and writing to shared (atomic) registers. 

\subsection{Failure Models}
Furthermore, the system can be susceptible to process failures.  Here, we consider two types of process failures.  The first one is the \textbf{Crash failures} where a process simply stops its execution. The second type is  the \textbf{Byzantine failure} where a process may arbitrarily deviate from its protocol specification. Note that a crash is a special case of Byzantine failure.
Assuming that runs are infinite a faulty process for crash failure make a finite number of steps.
A process is \emph{correct} in a run if it does not fail at any point of its execution.

\subsubsection{Encryption Scheme}
To ensure authentication, we employ a public key encryption scheme, where each process possesses a signing (private) key $sk_i$ and knows the public key $pk_j$ of every other process $p_j$. A process can sign a message $m$ using its private key as $\sigma = \text{sign}(sk_i, m)$. We assume a perfectly secure signature scheme, ensuring that no signature of a correct process can be forged. A process can also forward a received message from process $p_j$ by adding its own signature to the message.
\subsection{Two Useful Primitives}
In sections~\ref{sec:auth} and \ref{sec:RW} we present algorithms that solve $k$-set agreement. Our solutions rely mainly on  two primitives \textbf{Terminating Reliable Broadcast (TRB)} for Byzantine failure with authentication and the \textbf{Snapshot} primitive for the shared memory model.

\subsubsection{Terminating Reliable Broadcast}

 A  TRB protocol typically organizes the system into a sending process and a set of receiving processes, which includes the sender itself. A process is called 'correct' if it does not fail at any point during its execution.\\ The goal of the protocol is to transfer a message from the sender to the set of receiving processes, at the end of TRB a process will 'deliver' a message by passing it to the application level that invoked the TRB protocol.\\
In order to tolerate arbitrary failures, the TRB protocol is enriched with authentication so that the ability of a faulty process to lie is considerably limited, and also detected by correct processes; thus deliver a "sender faulty" message. This protocol then works for any number of faulty processes.

Consider a set of value $V$, and a special value $SF$ (for sender faulty).  
A TRB protocol is a protocol of broadcast value with a process  $p$  being the sender, and making a \emph{TRB-bcast(v,p)}  for some $v\in V$.  All the  correct processes deliver a value $m$ by a \emph{TRB-deliver(p)}, where $p$ is the sender of the signed message $m$, in such a way to satisfy the following properties:

\begin{itemize}
	\item \textbf{Termination.} Every correct process delivers some value. \label{ter-trb}
	
	\item \textbf{Validity.} If the sender, $p$,  is correct and broadcasts a message $m$, then every correct process delivers $m$.\label{val-trb}	\item \textbf{Integrity.}  A process delivers a message at most once, and if it delivers some message $m \neq SF $, then $m$ was broadcast by the sender. \label{int-trb}
	\item \textbf{Agreement}. If a correct process delivers a message $m$, then all correct processes deliver $m$. \label{agr-trb}
\end{itemize}
The main idea of the algorithm for solving the TRB is the following. If $p_0$ the sender wants to broadcast a value $m$,  it signs this value then send it. When $p_1$ a process that receives that message,it signs the received message and forwards it to the next process $p_2$ and so on
and so forth until a process $p_i$ receives that message, and it signs it. We represent such a message as $m:p_0:p_1:\ldots :p_i$. With $ m : p_0$ being the result of $p_0$ signing $m$.

When a correct process receives a message $m:p_0:p_1:\ldots :p_i$, this message should be valid before the process extracts $m$ from it. We say that a message is \emph{valid} if (i) all processes that
have signed the message are distinct, and has the form $m:p_0:p_1:\ldots :p_i$, note that the valid messages are the only one that 'count', in the sense that all the other non-valid messages are ignored by correct processes.

If $m:p_0:p_1:\ldots :p_i$ is valid:
\begin{enumerate}
	\item The process extracts the value $m$ from the message, then
	\item It relays the value if didn't do before, with its own signature appended
\end{enumerate}

At round $t+1$ :
if the process has extracted exactly one message, it delivers it otherwise it delivers \textbf{SF}. 

As a matter of fact TRB is solvable in synchronous models with Byzantine failures and authentication \cite{dolev1983authenticated,SrikanthT87}, and we recall \cref{TRB-broad}, where $p_0$ is the sender ensuring that.

\begin{algorithm2e}[!htb]

		\SetAlgoLined
		
		\underline{\emph{TRB-bcast($v_0$,$p_0$)}:}\\  	
		$	m := v_0 $\\ extracted := $m$\\
		
		\underline{\emph{TRB-deliver($m:p_0$)}}\\ 
		
		\Comment{--In Round 1--}
		sign $m$ and  $\ssend$ $m : p_0$ to all;\\

		\underline{\textit{At the end of round $t+1$}}\\ 					
		\eIf{$\exists m$  s.t extracted = $\{m\}$}{deliver $m$ }{deliver \textbf{SF}; } 
		
		\caption{Algorithm Of TRB, with $p_0$ being the sender.}
		\label{TRB-broad}

\end{algorithm2e}

And we recall \cref{TRBalgo}, where $p_0$ sends $v_0$.

\begin{algorithm2e}[H]
	
		\SetAlgoLined						
		
		\underline{\emph{TRB-deliver($p_0$)}}\\ 
		
		In round $i$, $1 \le i \le t+1$  ; \\
		\ForEach{signed message $s'   \in relay $}{ sign $s'$ and  $\ssend$ $s' : p$ to all}  
		
		Receive round $i$ messages from all processes ;\\
		relay := $\emptyset$ ;\\
		
		\ForEach{  valid message $s' = m:p_0:\dots:p_i$ received in  round $i$ }{ 
			\If{$m \notin$ extracted}{ extracted := extracted $\cup$ \{$m$\} ; \\  relay := relay $\cup$ \{$s'$\}  }
			
		}

		\underline{\textit{At the end of round $t+1$}}\\ 					
		\eIf{$\exists m$ $\in \mathcal{M}$ s.t extracted = $\{m\}$}{deliver $m$ }{deliver \textbf{SF}; } 
		
		\caption{Algorithm of TRB with sender $p_0$.\\ Code of a process $p$ with $p \neq p_0$.}
		\label{TRBalgo}

\end{algorithm2e}

\subsubsection*{Proof of TRB}

\begin{claim}
	If the sender is correct and broadcasts a message $m$, then every correct process delivers $m$.
	\label{clai:broaDel}
\end{claim}
\begin{proof}
	If the sender is correct and wants to broadcast $m$, by definition it extracts $m$ (and no other value) in 'round' 0. It does not extract any other value in any round $i > 0$, because to do so it would have received a valid message $m' : p_0 : \dots : p_i$ where $p_0 = sender$, which contradicts the unforgettable property of authenticated messages. Thus, if the sender is correct, it extracts only the message it broadcast.
\end{proof} 

\begin{claim}
	If a correct process extracts $m$
	then all correct processes will extract $m$.
	
\end{claim}
\begin{proof}
	Let  $i$ be the earliest round in which some correct process
	extracts $m$ and let $p$ be such a
	process.\\ 
	\textbf{Base case:}  If $i=0$, then $p_0$ is the sender, and it will send $m :
	p_0$ to all processes in round 1, and all other correct
	processes will extract $m$ in that round. Thus, all correct
	processes will extract $m$, as wanted.
	From \cref{clai:broaDel}, no correct process will extract a message from $p_0$ if $p_0$ did not send it.
	\\
	
	Thus, we may assume that $i>0$. Process $p$ extracts $v$ because
	it has received a valid message $m : p_0 : \dots : p_i$ in round
	$i$. By the definition of the valid message, $p_0, \dots, p_i$
	are all distinct.\\
	We claim that all the processes in the sequence $p_0, \dots, p_i$ are faulty. We suppose, for contradiction, that $p_j$ is correct, for some $j$ such
	that $1\le j \le i$. Since the signature of a correct process
	cannot be forged, it follows that $p_j$ signed and relayed the
	message  $m : p_0 : \dots : p_j$ in round $j$. Since $p_j$ is
	correct it extracted $m$ in round $j-1 < i $.  Which contradicts the assumption that $i$ is the earliest round a correct process extracts $m$. Thus, $p_0, \dots, p_i$ are distinct
	and faulty; hence $i\le t$. 
	
	Therefore, $p$ will send a valid message $m : p_0 : \dots :
	p_i : p$  to all processes in round $i+1 \le t+ 1$. 
	All correct processes will receive that message and will
	extract $m$ in round $i+1$ if they have not done so
	already. Thus, all correct processes extract $m$, as wanted. 
\end{proof}
From the claim, it follows that all correct processes extract the same set of values. Thus they all deliver the same message, proving Agreement.

The termination is trivial: The sender $p_0$  delivers in round 0, and every other correct process  will deliver a message at the end of the $t+1$ rounds

\begin{claim}
	A process delivers a message at most once, and if it delivers some message $m \neq SF $, then $m$ was broadcast by the sender.
\end{claim}

\begin{proof}
	At the end of the $t+1$ rounds, if a correct process delivers a message $m$, then the set of extracted contains only one message $m$.
	If a correct process delivers $m$, then $m$ was extracted from a valid message, signed and broadcast by the sender, in case of $p$ being the sender, it stores $m$ in extract in round 0, and delivers it.
\end{proof}

\subsubsection{Snapshot\label{append-snapshot}}

The snapshot object was introduced in~\cite{snapshot} as a shared data structure allowing concurrent processes to store
information in a collection of shared registers. It can be seen as an initially empty set, which can then contain up to $n$ values (one per process). This object provides two operations denoted $update()$ and $snapshot()$.  
The invocation $update(m)$ by a process $p_i$ writes in process $p_i$'s register the value $m$.
The invocation $snapshot()$ by a process $p_i$ reads process $p_j$'s register and returns its content, which we denote as $view_i$.



We consider an atomic snapshot object  i.e the snapshot object satisfies the following properties:
\begin{itemize}
	\item \textbf{Termination.} The invocation of $snapshot()$ or $update()$ by a correct process terminates.
	 \item \textbf{Atomicity Property.} The snapshot  and update operations are  atomic, meaning that they appear to execute instantaneously and they satisfy the sequential specification of the snapshot. 
	
\end{itemize}

By enforcing Atomicity on the snapshot operation, each process can obtain a consistent and unchanging view of the shared registers, preventing any concurrent modifications that could lead to data corruption or incorrect decisions. 

The snapshot object satisfies the following property when the code of each  process is  first the invocation of update for some value (its local value)  then any number of invocation to snapshot.
\begin{itemize}

	 \item \textbf{Inclusion Property.} (1) When a process $p_i$ takes a snapshot, the resulting view $view_i$ includes the local value of $p_i$.
	 (2) For any process $p_j$, if the snapshot of $p_i$ occurs before the snapshot of $p_j$, then $view_i$ is a subset of $view_j$.
	
\end{itemize}

\subsection{$k$-Set Agreement}
In $k$-set agreement, each process must decide on a value such that no more than $k$ different values are decided by correct processes. More precisely, 
let $V$ be a finite set of  at least $k+1$ values. Each process has an
\emph{initial} value $v$ in $V$ and we say that {\it proposes} $v$ to $k$-set Agreement.
Each process has to irrevocably {\it decide} on a value in $V$. The decided values must satisfy the following properties. 
\begin{itemize}
	\item \textbf{Validity.}
	If all the correct processes propose the same initial value $v$,
	no correct process decides a value different from $v$.   
	\item\textbf{ Agreement.}
	At most $k$ different values are decided by the correct processes. 
	\item \textbf{Termination.} Eventually, all the correct processes decide.
        \end{itemize}

 We say that an algorithm solves $k$-set Agreement in a system of $n$ processes
with at most $t<n$ failures of processes, if 
all the executions in this system satisfy these properties. When $k=1$, $k$-Set Agreement is the classical \emph{consensus}.     
Several non equivalent validity properties for $k$-set agreement have been  proposed and argued  \cite{AsynchKset}. The validity considered here is generally the one used in the Byzantine case. 
More recently, \cite{CivitGGKV23} argues on the possibilities and impossibilities of various validity in the context of consensus.
 
Remark that the validity property given here is stronger than the validity property generally given for crash failures in which a decided value has only to be one of the initial of processes (correct or faulty). Hence, in some case, a decided value would come from a faulty process, that can be acceptable when processes are not malicious as with crash failures. Moreover, in the Byzantine case, it is not clear what is the initial value of a Byzantine process (any value, no value?), and such a weak validity condition would lead to decide any value in all cases.

Note that, if applied to crash failure model, it could also be interesting from a practical point of view to  force the decided value when all correct processes propose the same value.

\subsection{A Lower Bound For Number For $k$-Set Agreement}
We demonstrate that the lower bound for the $k$-set agreement is $k = \lfloor\frac{n}{n-t}\rfloor$, with $k$ an integer, which implies that no lower value of $k$ is achievable in the system, with $k> 1$, when we have $t$ processes that fail, let us proceed with the formal proof.

\begin{theorem}
	$\lfloor\frac{n}{n-t}\rfloor$ is a lower bound for solving $k$-set agreement
	\label{th:lb}
\end{theorem}

\begin{proof}
	
	We consider a system with $n-t$ processes, and $t$ faulty. Let us partition the processes in $k = \lfloor\frac{n}{n-t}\rfloor$ subsets  $g_1, g_2, \ldots, g_{\frac{n}{n-t}}$ of size at least $n-t$.\\
	
	Let us consider a run $\alpha$ where in each $g_i$ all processes have initial value $v_i$, and there is at least a non-faulty process $p_i$, and all the faulty processes crash after all the correct processes have decided, (we can also consider that there is no crash).
	Moreover, let us suppose that the values $v_i$ are pairwise distinct.
	
	Let us consider a run $\alpha_{1}$, where all the processes in $g_1$ are correct and have initial value $v_1$. All the processes in $g_j$, for $j\neq 1$; have initial value $v_j$ and crash after all the correct processes have decided. As the processes in $g_1$ need to ensure validity they need to decide $v_1$.\\
	Notice that process $p_1 \in g_1$ cannot distinguish between $\alpha$ and $\alpha_{1}$. Thus, for any decision algorithm it has to decide $v_1$ in both runs.
	
	Generalizing this argument for any $p_i$, at least $k = \lfloor\frac{n}{n-t}\rfloor$  different values are decided in $\alpha$ by any decision algorithm.
\end{proof}

 As we give in section~\ref{sec:auth} an algorithm for $k$-Set
 Agreement in the Byzantine case with authentication for  $k \geq
 \frac{n}{n-t}$, we deduce that the $k \geq \frac{n}{n-t}$ is a lower
 bound to get $k$-set Agreement.


\section{$k$-Set Agreement in an authenticated  Synchronous Message
  Passing Model \label{sec:auth}}
\label{sec:cont}

In this section, we explore $k$-Set agreement in a message-passing model
with authentication.
In fact all these results apply to crash failures models too.
Our focus is on ensuring reliable communication between processes
through the exchange of \textbf{authenticated} messages in the
TRB algorithm. Each received message is signed by the sender, guaranteeing its
authenticity. Furthermore, we assume that the communication is
reliable, meaning that messages are neither lost, forged, nor
generated by the network.

In a first step, we present a $k$-set agreement algorithm in two
rounds, with no constraints on the number of failures $t$. But with a
value of $k$ such that $k > \lfloor \frac{n}{n-t} \rfloor $ which
is not an optimal value for $k$. In a second step we
give an optimal algorithm concerning the value of $k$ for $k \ge \lfloor \frac{n}{n-t} \rfloor $, where this algorithm needs $t+1$ rounds. 


\subsection{A Two Rounds $k$-Set Agreement  In a Byzantine Failures Synchronous Model With Authentication}
We first present a two rounds algorithm, that ensures $k$-Set Agreement
for $k= \lfloor \frac{n}{n-t} \rfloor +1$, this is an authenticated
algorithm where the messages sent by processes are signed.  This
prevents any faulty process from forging the signature of a correct
process or misrepresenting the value sent by a correct process.

\subsection*{The Algorithm}

Algorithm \ref{2roundBSMP} ensures $k$-set agreement for $k= \lfloor \frac{n}{n-t} \rfloor + 1$ and $\forall t$. The processes exchanges their messages in two rounds. In the first round, a process $p_i$ sends its initial value $v_i$, and receives every other process' initial value if any, it stores them in a vector $V_i$, and in the second round it sends this vector to every other processes and receive from every process $p_j$ a vector $V_j$ if any.


\begin{algorithm2e}[]
	
	\SetAlgoLined
	\KwIn{$v_i$  initial value}
	\KwResult{decide}
	
	\Comment{Local variables}  \label{phase1}
	$V_i$ Vector of decision of size $n$ initialized to $\perp$\\
		$M_i[i][i] = V_i[i] \leftarrow    v_i  $ \label{broadcast2round}; \\
{\bf for all} $r,\ell \neq i$ {\bf do} $M_i[j][r]\leftarrow \bot${\bf end for}; \\
	
	\BlankLine 
\Comment{----------------------- round 1  ----------------------------------------}

	\BlankLine
	Send the value $M_i[i][i]$ to all the processes ; \\
	\textbf{When} $p_i$ receives $v_j$ from $p_j$ \textbf{do} $M_i[i][j] \leftarrow v_j$; \\ 

	\BlankLine
	\Comment{----------------------- round 2  ----------------------------------------}
	
	 Send  the vector  $V_i = M_i[i][*]$  to all the processes;\\
	
	\textbf{when } $p_i$ receives a vector $V$  from $p_j$ 
	\textbf{do}  $M_i[j]\leftarrow V$;         \\
	
		\Comment{-------------------  at the end of round 2  -----------------------------------}
\For{$j=1$ \KwTo $n$, with $j\neq i $ }
{   
	$w = M[i][j]$;\\
	
	\uIf{ $w =\perp$}{ $V_i[j]= \perp$;}
	\Else{
		$V_i[j]= w $;\\
		\For{$\ell =1$ \KwTo $n$, with $\ell \neq i $}
	{ 	\uIf{ $M_i[\ell][j] \neq \perp $  }{  
			\If{$M_i[\ell][j] \neq w $}{$V_i[j]= \perp $ \label{botComp}}
		
		} 
	 }
	
	}

}

		\Comment{------------ Decision at the end of round 2  ---------------------}
	  \If{$p_i$ finds in $V_i$, $n-t$ values equal to $v_i$ }{  
		decide \textit{$v_i$;} \label{Dinit2round}
	} 
	\Else{
		decide $\perp$ \; \label{bottom2round}
	}

	\caption{Solving $k$-set Agreement  in
          synchronous message passing with Byzantine failures and
          authentication.
	\label{2roundBSMP}}
\end{algorithm2e}

\subsection*{Proof of the algorithm}
First, note that the exchanged messages are authenticated, where a Byzantine process might choose to not  relay a message, but it is constrained from misrepresenting or altering the content of a received message.  And only the valid signed messages are considered.

Each process $p_i$ manages a local matrix  $M_i[1..n][1..n]$,  such that
$M_i[i][i]$ is initialized to the value proposed by $p_i$, where  $M_i[i][i]$ is a signed message, and
all the other entries are initialized to $\bot$. 
In the first round each $p_i$ sends $M_i[i][i]$
to all the processes and assigns to  $M_i[i][j]$
the value it receives from $p_j$.

Then at the second round, each process $p_i$ broadcasts its
$M_i[i][\ast]$ vector, and the received vectors from other processes
are used to update the matrix $M_i$ as follows: $p_i$ receives from
$p_j$ the vector $M$ that it stores it in $M_i[j] [\ast]$.

An example of  $p_i$'s matrix $M_i$ is represented in  \cref{matrix mi}.  At the conclusion of this round, process $p_i$ compares the values of each column. If it detects that a process has sent two different values to two different processes, it sets its own values in $V_i$ to $\perp$. To guarantee the Validity property, if $p_i$ finds $n-t$ values in $V_i$ that are equal to its own initial value $v_i$, it decides on $v_i$ as the agreed value. Otherwise, it decides on $\perp$.
\begin{figure}[htb!]
	\centering
	\[M_i = 
	\begin{tikzpicture}[baseline=(math-axis),every right delimiter/.style={xshift=-3pt},every left delimiter/.style={xshift=3pt}, arrow/.style = {thick, -stealth}]%
		\matrix [matrix of math nodes,left delimiter=(,right delimiter=)] (matrix)
		{
			M_{11} && \cdots && M_{1i} && \cdots && M_{1j} && \cdots &&  M_{1n} \\
			\vdots &&        && \vdots &&        &&\vdots  \\
			|(Mi1)| M_{i1} && |(Mi2)| \cdots && |(Mi3)| M_{ii} = v_i && |(Mi4)| \cdots && |(Mi5)| M_{ij} && |(Mi6)| \cdots && |(Mi7)|  M_{in} \\
			\vdots &&        && \vdots &&        && \vdots  \\
			|(Mj1)| M_{j1} && |(Mj2)| \cdots && |(Mj3)| M_{ji} && |(Mj4)| \cdots && |(Mj5)| M_{jj} = v_j && |(Mj6)| \cdots && |(Mj7)|  M_{jn} \\
			\vdots &&        && \vdots &&        && \vdots  \\
			M_{n1} && \cdots && M_{ni} && \cdots && M_{nj} && \cdots &&  M_{nn} \\
		} ;
		\node (i)[color = red, draw,dashed,inner sep=0pt,fit=(Mi1) (Mi7)] {};
		\node (s1) at (4,2) {$M_i[i] [\ast]$};
		\draw [arrow, red] (s1) to (i);
		
		\node (j)[color = blue, draw,dashed,inner sep=0pt,fit=(Mj1) (Mj7)] {};
		\node (s2) at (4,-2) {$M_i[j] [\ast]$};
		\draw [arrow, blue] (s2) to (j);
		
		\coordinate (math-axis) at ($(matrix.center)+(0em,-0.25em)$);
		
	\end{tikzpicture}\]
	\caption{Matrix $M_i$}
	\label{matrix mi}
\end{figure}

\begin{lemma}
	Let $p_i$ be a correct process with matrix $M_i$, and let  $p_\ell$ and $p_j$ be two processes. If  $M_i[i][j] \neq M_i[\ell][j]$, which means that if process $p_j$ sent two distinct values, different than $\perp$ to $p_i$ and $p_\ell$, then $p_j$ is Byzantine.
\end{lemma}
 
 \begin{proof} \label{byzproc}
 	By contradiction, lets suppose that $p_j$ is correct and send the same value to everyone, thanks to the authentication, $p_j$'s message cannot be forged or lied about, then all the processes in $\Pi$ will have the same view on $p_j$' value, and all the processes correct or Byzantine will relay the same message sent by $p_j$.
 	
 \end{proof}

\begin{lemma}
	[Validity] If all the correct processes propose the same value, they decide that value.
\end{lemma}

\begin{proof}
 	Consider a set $C$ consisting of correct processes that propose their initial value $v$. Since we have at most $t$ faulty, $C$ is consisting at least of $n-t$ correct processes. Since the messages are authenticated, no Byzantine can lie about a correct process value, thus all the column corresponding to correct process have the same value or bottom in case a Byzantine lies about receiving nothing. For every process $p_j$ in $C$, $V[j] = v$. Thus, by the end of the second round, all correct processes will have at least $n-t$ values in their vector that are equal to $v$. Then, in  \cref{Dinit2round} of the algorithm, when a correct process observes this condition, it will decide on the value $v$.
\end{proof}

\begin{lemma}
	[Agreement] at most $\lfloor \frac{n}{n-t}\rfloor +1$ values are decided.
\end{lemma}

\begin{proof}
We first prove that no more than $\lfloor \frac{n}{n-t}\rfloor$  different values  from $\perp$ are decided.

First, note that a Byzantine process may lie about another Byzantine process' value, in this case we have two different values in the same column, then a correct process analyzing its matrix will set that value to $\perp$.

Let $p_\ell$ be a Byzantine process that sends two distinct values, to two correct processes  in the first round, since the line of the correct processes are equal in all the matrices, all the correct processes will detect that $p_\ell$ is Byzantine, and set its value to $\perp$ in \cref{botComp}.

Now we suppose that $p_\ell$ does not send two distinct values to two correct processes in the first round, but send $v_\ell$ to one and nothing to the other. In this case, at the end of the second round, let us suppose two correct processes $p_i$ and $p_j$ with respectively $V_i$ and $V_j$ as decision vector. Either the two processes have $V_i[\ell] = V_j[\ell] $ or one of them has $\perp$ and the other $v_\ell$.

From all the above, the wort case scenario is when $p_\ell$ a Byzantine process is acting like a correct, in this case, when a correct process $p_i$ is deciding its initial value, the occurrence of $v_i$ in $V_i$ has to be $n-t$. Thus, we can have up to $\lfloor \frac{n}{n-t}\rfloor$  groups of size $n-t$ with different values, since a Byzantine process cannot belong to a set of processes proposing $v$ and another set of processes proposing $v'$.  

In conclusion, we have at most $\lfloor \frac{n}{n-t}\rfloor$  decided values plus the $\perp$.

\end{proof}

\begin{lemma}	
	[Termination] All the correct processes decide.
\end{lemma}

\begin{proof}
	All the correct processes will execute the two rounds, and decide at the end of the second round.
\end{proof}

By the above lemmas, we get the following theorem.

\begin{theorem} \label{th:2bsmp}
	Algorithm \ref{2roundBSMP}, ensures $k$-set agreement in an authenticated message passing system with Byzantine failures for $k= \lfloor \frac{n}{n-t}\rfloor +1$
      \end{theorem}
      


\subsection{An Algorithm For  $k$-Set Agreement In Byzantine Failures
  Synchronous Message Passing Models: Optimal Concerning The Value $k$}

In this section, we present an algorithm for $k$-set Agreement when $k
\ge  \lfloor \frac{n}{n-t} \rfloor$. From Lemma~\ref{lowerbound} it is
the best we can do and then this algorithm is optimal concerning $k$.

Before that, we give a Terminating Reliable Broadcast~\cite{ SrikanthT87} that is in the
heart of the following $k$-set Agreement algorithm.

For this, we use TRB to solve an interactive
consistency~\cite{PSL80}, namely correct processes will agree on a
$n$-vector  corresponding to the initial values of processes.
More precisely, each process $p_i$ invokes first a $TRB-bcast(v,p_i)$
with its initial value $v$. Then each process $p_i$ will fill a local
$n$ vector $L_i $ with the values obtained by each $TRB-deliver(p_j)$.
It is the Phase 1 of algorithm~\ref{BSMP}.

Specifically, we introduce an authenticated algorithm that ensures
$k$-set agreement for $k = \lfloor \frac{n}{n-t} \rfloor$. To achieve
this level of agreement, we are constrained to use more than two
rounds. The algorithm incorporates a primitive that operates over
$t+1$ rounds. Once again, there are no limitations on the number of
failures $t$. This algorithm uses $n$ instances of the
\textbf{Terminating Reliable Broadcast } (TRB), where the delivered
value is the proposed value for the set agreement.  

For our algorithm, we implement the authenticated TRB primitive
for $n$ instances, where the $ith$ instance of TRB corresponds to the
run instance where process $p_i$ is the sender in TRB. 

\subsection*{The Algorithm}

Algorithm \ref{BSMP}, ensures $k$-set agreement for $k \ge  \lfloor \frac{n}{n-t} \rfloor$ and $\forall t$.
when $p_i$-$TRB-bcast(v_i)$  is called then  process $p_i$ is the sender in the TRB algorithm, it stores its value in the vector $L_i$, while on the call of $q$-$deliver(m)$, a process $q$ is the sender in the TRB algorithm that sent a value $m$, and $q_i$ will deliver that value in the TRB algorithm.


\begin{algorithm2e}[!htb]
	\SetAlgoLined

		
		\Init{$L_i$ vector initialized to $perp$}

		\Comment{------------------ Phase 1  ----------------------------------------}

		\underline{ \emph{TRB-bcast($v_i$,$p_i$):}}\\
		
		\label{broadcast4}  $L_i[p_i] \gets v_i$ ;\\
		
		\ForEach{$q \in \Pi$}{
			\underline{ \emph{TRB-deliver($q$)}:}  \\   
			
			\label{deliver4} $L_i[q] \gets m$; \\
			
		}

		\Comment{------------------ Phase 2 --------------------------}
		
		\uIf{$p_i$ finds in $L_i$, $n-t$ values identical to $v_i$}{
			decide $v_i$; \label{decide4} \;
		}
		\uElseIf{$p_i$ finds a value $v$ repeated $n-t$}{
			decide $v$;
		}
		\Else{
			decide $\perp$; 
		}

	\caption{Solving $k$-SA in Byzantine Synchronous Message passing model.\\ Code for $p_i$.}
	\label{BSMP}
	
\end{algorithm2e}

\subsection*{Proof of the algorithm}
Every process $p_i$ holds a vector $L_i$ initialized to $\perp$.  $n$ instances of TRB are started, one per process, with each process $p_i$ being the sender in one instance $i$.

Every process $p_i$ records in vector $L_i$ the message $m$ delivered from process $q$. We consider $L_i$ the vector of the proposed values for the $k$-set Agreement. 

At the end of Phase 1, after each process sets its vector $L_i$, even if we have in the system Byzantine processes, all the vectors of correct processes will be equal, and if all the correct processes propose the same value, we  will have at least $n-t$ values equal to their initial value. \\
This ensures the Validity of $k$-set agreement,  at the end of Phase 2, they will decide that value in  \cref{decide4}. 

Let us suppose we are in an execution where the correct processes do not have the same initial value, since we have at least $n-t$ correct processes, for every process $p_i$ we have at least $n-t$ values different from $SF$ in $L_i$.\\

At the end of the $n$ instances of TRB, all the correct processes have the same  $L$. Thus, in phase 2, each correct process $p_i$, decides from $L_i$, by  either finding  $n-t$ values equal to its initial value, or any value repeated $n-t$ times, if not  decides $\perp$.
We will see in the following lemma prove that if a process decides $\perp$, then no other correct process decides a value different from $\perp$.

\begin{lemma} 
	Let $p_i$ and $p_j$ be two correct processes that run  \cref{BSMP}. $L_i = L_j$,  at the end of Phase 1. 
	\label{equalL}
\end{lemma}

\begin{proof}
	From the  agreement of the TRB [\cref{agr-trb}], we have if a correct process delivers a value $m$, that value is delivered by all the correct processes, the delivered message is stored in the vectors of correct processes, then all the correct processes will have the same value in their vector.
	
\end{proof}

\begin{lemma}
	The vector $L$ contains the initial value broadcast by correct processes.
	\label{intivL}
\end{lemma}

\begin{proof}
	Let $L$ be the vector $L_i$ of a correct process $p_i$, from \cref{equalL}, all the correct processes have the same $L$. 
	Since in Phase 1, every correct process  $p_i$ stores its initial value in $L_i$ in \cref{broadcast4}, then every initial value broadcast by a correct process $p_i$ is in $L$.
	
\end{proof}
Hence, from the above Lemmas  we conclude that at the end of phase 1, all the correct processes have the same view on $L$.
\begin{lemma}
	If a correct process $p_i$ decides a value $v$ different from $\perp$, no other correct process will decide $\perp$. 
	\label{DecBottom}
\end{lemma}

\begin{proof}
	Let $p_i$ be a correct process that decides on the value $v$ in Phase 2. This implies that $p_i$ has found at least $n-t$ instances of the value $v$ in $L_i$. Consider another correct process $p_j$. From \cref{equalL}, given that the vectors $L_i$ and $L_j$ are equal for all correct processes, $p_j$ will also find at least $ n-t $ instances of the value $v$. Thus, every correct process will at least decide $v$.
	
\end{proof}

\begin{lemma}
	{\em [Validity]} If all the correct processes propose the same value,
	they decide this value.
	\label{lemVal}
\end{lemma}

\begin{proof}
	Let $C$ be a set of correct processes that propose $v$, from  \cref{equalL,intivL} we have at least $n-t$ values equal to $v$ in $L$, then in  \cref{decide4}, all the correct processes will find that condition and decide $v$.
	
\end{proof}

\begin{lemma}
	{\em [Termination]} All the correct processes decide.
	\label{sim:term}
\end{lemma}

\begin{proof}
	All the correct processes will execute the two phases and will decide at the end of Phase 2.
\end{proof}

\begin{lemma}
	{\em [Agreement]} At most $\lfloor \frac{n}{n-t} \rfloor $ values are decided by correct processes.
	\label{sim:agre}
\end{lemma}

\begin{proof}
	
	If a correct process decides $\perp$, all the correct processes will decide that value. Thus, exactly one value is decided.
	
	If a correct process decides a value different from $\perp$, by \cref{DecBottom}, all correct processes decide a value that have a frequency of at least $n-t$.
	
	We suppose that the correct processes decide on a specific number of distinct values $\alpha$.
	Given that each of the $\alpha$ values must appear at least $n-t$ times in $L$ for a decision to be made, the total count of these appearances is $\alpha(n-t)$. This count cannot exceed the total number of the process in the system, we have $\alpha(n-t) \le n$, leading to  $\alpha \le \lfloor \frac{n}{n-t} \rfloor $.
	From \cref{equalL}, no correct process will decide $\perp$.

\end{proof}

\subsubsection{Case of Consensus in crash failures or Byzantine failures with authentication}

When examining Consensus with the same validity condition, where the decided value must be the same if all correct processes propose the same value, the traditional bound of $t+1$ rounds to achieve consensus \cite{aguilera1999} applies. Building upon the result from the previous section, we have:
\begin{theorem}
  For crash failures or Byzantine failure with authentication, $1$-set Agreement (consensus) is solvable if and only if there is a majority of correct processes. Moreover the algorithm requires $t+1$ rounds.
\end{theorem}





\section{$k$-Set Agreement in a Crash failures Asynchronous Read-Write shared memory models} \label{sec:RW}
In this section, we present an algorithm that operates in an
Asynchronous Read-Write (Shared Memory) setting, specifically designed
to handle crash failures. This algorithm is applicable to a system
consisting of $n$ processes, among which up to $t$ processes may
crash.
It is known that in the asynchronous system, there is no algorithm for solving consensus when the system is subject to crash failures\cite{fischer1982impossibility}. The authors in \cite{AsynchKset} investigates the $k$-set agreement in an asynchronous system, exploring several variations of the problem definition by varying the validity condition and the system model.
In our case, we are interested in the following validity: \textit{if all the correct processes propose the same value, that value is decided}, and the shared memory model, where they presented an impossibility result. 

\begin{theorem} \cite{AsynchKset}
	In the Shared Memory/Crash model, there is no protocol for
	solving $k$-set agreement when $t \ge \frac{n}{2}$ and $t \ge k$.
\end{theorem}

And a possibility one: 

\begin{theorem} \cite{AsynchKset}
	There exists a protocol that can solve $k$-set agreement for $t < \frac{k-1}{2k}n$.
\end{theorem}

They left a small gap between their possibility and impossibility results. In this section, we managed to find an algorithm for the case for $ t < \frac{k-1}{2k-1}n$, which is equivalent to $k > \lfloor \frac{n-t}{n-2t} \rfloor$.
Our algorithm uses the snapshot primitive, in a look-alike first phase of a round then exploiting the result returned by the snapshot each process makes a decision.

\subsection*{The Algorithm}
Algorithm \ref{CSRW} solves $k$-set agreement for $k > \lfloor \frac{n-t}{n-2t} \rfloor$, since we are in an asynchronous system, processes can not wait indefinitely, as a process cannot distinguish between a correct process that is slow and a process that crashed. Thus if a process receives at least $n-t$ values it moves to a decision-making step.


\begin{algorithm2e}[]
	
	\SetAlgoLined
	\KwIn{ Initial value $m$}
	\KwResult{decide}
	\Comment{Shared  variables}  
	$S$ : snapshot \\
	\BlankLine
	\Comment{Local variables}  
	
	$X_i $and $L_I$ are two vectors $[1,\dots,n]$ initialized to $\perp$ \\
	$x_i = 0$
	\BlankLine
	$S.update(m) $\label{write} \\
	\BlankLine
	\While{  $|  j, L_i[j] \neq \perp  |  < n-t$ \label{snapshot}}{$L_i \leftarrow     S.snapshot()$\\}
	$X_i \gets L_i$ \label{snapX}\\
	
	$x_i \leftarrow \sharp$  of values $\neq \perp$ in $X_i$ \label{nbr-val}
	
	\BlankLine

	\uIf{$p_i$ finds in $X_i$, $x_i-t$ values equal to initial value $m$ }{  
		decide \textit{$m$;} \label{decideinit}
	}  
	\uElseIf{$p_i$ finds in $X_i$, $x_i-t$ values equal to any value $v$}{decide $v$; \label{decval}} 
	\Else{
		decide $\perp$ \; \label{bottom}
	}

	\label{CSRW}
		\caption{Solving $k$-set Agreement  in
	Crash failures shared memory models.}
\end{algorithm2e}

\subsection*{Proof of the algorithm}

We begin with a snapshot object denoted as $S$. 
Each process $p_i$ starts by updating the snapshot $S$ with its initial value  $m$ by invoking the operation $update(m)$.  It then calls the function $snapshot()$ and stores the resulting view in a vector $X_i$ of size $n$, initially set to $\perp$. The processes continue invoking the snapshot function until $n-t$ processes have updated the snapshot.

At line \ref{nbr-val}, process $p_i$ calculates $x_i$, the number of values in $X_i$ that are different from $\perp$,and it has to be $\ge n-t$ because of the condition in line \ref{snapshot}. If $p_i$ finds $x_i-t$ values equal to its initial value, it decides on that value. If not, it checks if there are $x_i-t$ values equal to a value $v$ and makes a decision based on that value. If neither of these conditions are met, it decides on $\perp$ as the outcome.

It is important to note that due to the inclusion property,  if a snapshot $X_i$ is taken at an earlier point in time than another snapshot $X_j$, then the values represented by $X_i$ is included in the values represented by $X_j$, and we say that $X_i$ is smaller that $X_j$.
If the process $p_i$ decides $v$ it has to  find at least  $x_i-t $ values $v$ in $X_i$. By the inclusion property, the process $p_j$ finds also at least $x_i-t$ values $v$ on $X_j$, but to decide $v$, it has to find $x_j-t $ values $v$.
%
%
%
%

\begin{lemma}
	Let $p_i$ be a correct process that decides a value $v$, $v\neq \perp$, for every correct process $p_j$ with $X_i> X_j$, $p_j$ will not decide  $\perp$. 
	\label{decnotBot}
\end{lemma}

\begin{proof}
	Let $p_i$ be a correct process that decides $v$. Let $p_j$ be a correct process with $X_j<X_i$, lets prove that $p_j$ will not decide $\perp$. Let $\alpha_v$ be the number of occurrences of $v$ in $X$.
Since $p_i$ decided $v$, then $\forall v \in X_i$ ,  $\alpha_v^i \ge x_i - t $.
  Since $X_j<X_i$

\begin{equation}
	\begin{split}
		& 	\alpha_v^i - \alpha_v^j \le x_i - x_j  \Rightarrow  \alpha_v^i - x_i \le \alpha_v^j - x_j \\
		&  \text{ Since }  \alpha_v^i \ge x_i - t \\
		& \text{We obtain }  \alpha_v^j - x_j \ge -t \Rightarrow   \alpha_v^j \ge x_j -t  
	\end{split}
\end{equation}
That does not necessary implies that $p_j$ will decide $v$, since it looks first for its initial value, but as there is at least $\alpha$ values of $v$ in $X_j$ such that $ \alpha_v^j \ge x_j -t$, $p_j$ will not decide $\perp$.  
	
\end{proof}

Let $X_i$ be the smallest snapshot among all the snapshots in \cref{snapX}.
\begin{lemma}
	No correct process decides a value $v$, if  $\alpha_v^i < x_i - t$.
	\label{occuS}
\end{lemma}

\begin{proof}
	
	Let $p_j$ be a correct process with $X_j \ge X_i$.
	Let us suppose by contradiction that  $p_j$ decides $v$. Then $\alpha_v^j \ge x_j - t$.
	
	Since $X_i \le X_j$:
	\begin{equation}
		\begin{split}
			& 	\alpha_v^j - \alpha_v^i \le x_j - x_i  \Rightarrow  \alpha_v^j - x_j \le \alpha_v^i - x_i \\
			&  \text{ Since }  \alpha_v^i  < x_i - t \\
			& \text{Then }  \alpha_v^j - x_j < -t \Rightarrow   \alpha_v^j < x_j -t  
		\end{split}
	\end{equation}
	
	Contradicting the fact that   $\alpha_v^j \ge x_j - t$.
\end{proof}

\begin{lemma}
	[Validity] If all the correct processes propose the same value, that value is decided.
\end{lemma}

\begin{proof}
	If all the correct processes propose the same value $v$, then there are at most $t$ values ( proposed by the faulty processes) that can be different from $v$.
	So for a correct process $p_i$, in its snapshot $X_i$ there are at most $t$ values different from $v$, therefore $x_i-t$ values of $X_i$ are equal to $v$ and $p_i$ decides $v$ Line  \ref{decideinit}.
\end{proof}

\begin{lemma}
	[Agreement] At most $k$ values decided with $k > \lfloor \frac{n-t}{n-2t} \rfloor$.
\end{lemma}

\begin{proof}
Let $X_i$ be the smallest snapshot among all the snapshots in \cref{snapX}.
From \cref{occuS}, a value $v$ can only be decided  if $\alpha_v^i \ge x_i -t$.

In $X_i$, there can be at most $\lfloor \frac{x_i}{x_i - t}\rfloor$ different values $v$ such that $\alpha_v^i \ge x_i -t$.

As $n-t \le x_i \le n$ and $\lfloor \frac{x_i}{x_i - t}\rfloor$ is a decreasing function, the maximum number of different values decided is obtained when $x_i = n-t$, thus, $\frac{n-t}{n-2t}$ is the maximum number of decided values.
\end{proof}

\begin{lemma}
	[Termination] All the correct processes will eventually decide.
\end{lemma}

\begin{proof}
	The number of failures is at most $t$, we have then at least $n-t$ correct processes that update the snapshot  $S$ so 
	a correct process will find a snapshot of $S$ of size at least $n-t$ and decides
	
\end{proof}

By the above Lemmas, we get the following Theorem.

\begin{theorem} \label{th-csrw}
	For $k >\frac{n-t}{n-2t}$, 
	Algorithm~{\em{\ref{CSRW}}} ensures $k$-set agreement.
\end{theorem}

\begin{theorem}	
	There is no Algorithm that can solve $k$-set agreement for $k\le \lfloor \frac{n-t}{n-2t} \rfloor - 1$.
\end{theorem}

\begin{proof}

Let us partition $n-t$  processes into $ \lfloor\frac{n-t}{n-2t}\rfloor$ subsets  $g_1, g_2, \ldots, g_{\frac{n-t}{n-2t}}$ of size of at least $n-2t$ processes, and let $t$ be the remaining processes.

We consider a run $\alpha$, where for each $g_i$, at least a correct process $p_i$ proposes its initial value $v_i$. Let $\tau$ be the time at which all the correct processes decide in $\alpha$. Before $\tau$, all the processes in $g_i$ write $v_i$, while the remaining processes $t$ do not take any step.

Let $\alpha_1$ be a run, where after $\tau$, for each $g_i$, with $i>1$, the processes in $g_i$ crash, and the remaining $t$, have an initial value $v_i$ and send it after $\tau$. 
All the correct processes in $g_1$, to ensure validity have to decide $v_1$.

A correct process $p_1$ in $g_1$ cannot distinguish between $\alpha$ and $\alpha_{1}$, thus decides the same value $v_1$ in both runs.

Generalizing the same argument,t for every $g_i$ in $\alpha$. There are at least $ \lfloor\frac{n-t}{n-2t}\rfloor$ different decided values.

\end{proof}

\subsection{$k$-Set Agreement in Crash Failures Asynchronous Message Passing Model}
It is shown in \cite{emulation}, that any wait-free algorithm based on atomic, single-writer (and multi-writer) multi-reader registers can be automatically emulated in message-passing systems, provided that at least a majority of the processors are not faulty. We have $ t < \frac{k-1}{2k-1}n$ then we have majority of correct since $\frac{2k-1}{k-1} > 2 $.
In order to simulate a  message-passing model with a shared memory
one, we consider that there exists a channel from a process $p$ and
$q$ where $p$ is the writer and $q$ is the reader. And a channel from
$q$ to $p$, where $q$ is the writer and $p$ is the reader. 

From  Theorem \ref{th-csrw}, we have the following Theorem.

\begin{theorem} \label{th-camp}
For $k >\frac{n-t}{n-2t}$, we have an algorithm ensuring $k$-set agreement in an asynchronous message passing model.
\end{theorem}


\section{Conclusion}
In this study, we have contributed to the understanding of $k$-set agreement in distributed systems, particularly focusing on scenarios with different types of failures. We addressed and filled certain gaps left by previous works such as \cite{AsynchKset} and \cite{DBLP:conf/netys/Delporte-Gallet20}, providing valuable insights into the solvability of the problem.

In the context of the Synchronous Message Passing model with Byzantine failures, we presented two algorithms that achieve $k$-set agreement. The first algorithm ensures $k = \lfloor \frac{n}{n-t} \rfloor + 1$ in just two rounds, demonstrating its efficiency and practicality. On the other hand, the second algorithm achieves optimality, ensuring $k \ge \lfloor \frac{n}{n-t} \rfloor$ in $t+1$ rounds.

We also present an algorithm for the Asynchronous Shared Memory model with crash failures, which solves $k$-set agreement for $k > \frac{n-t}{n-2t}$.

Furthermore, we extend these results by leveraging the equivalency of the models. The Synchronous Authenticated Byzantine model is equivalent to the Synchronous Crash model, and the Shared Memory model canbe simulated with the message passing model.

Although certain gaps remain open, our aim is to provide answers and establish a comprehensive understanding of the solvability of $k$-set agreement in both asynchronous and synchronous systems. We emphasize the importance of the validity property, which guarantees that if all correct processes propose the same value, that value will be decided upon.

\bibliographystyle{ACM-Reference-Format}
\bibliography{biblio}

\end{document}